\documentclass[runningheads]{llncs}
\usepackage[T1]{fontenc}

\usepackage{amssymb,amsfonts}
\usepackage[ruled]{algorithm2e}
\usepackage[hidelinks]{hyperref}

\newcommand{\alg}[1]{\textnormal{\scshape #1}}
\usepackage{graphicx}
\begin{document}
\title{Online Busy Time Scheduling with Flexible Jobs}
%
%
\author{Susanne Albers\and G. Wessel van der Heijden}
\authorrunning{S. Albers and G. W. van der Heijden}
%
\institute{Technical University of Munich, Germany; TUM School of Computation, Information and Technology, Department of Computer Science  \\\email{\{susanne.albers,wessel.heijden\}@tum.de}}
\maketitle              
\begin{abstract}
We consider the online busy time scheduling problem motivated by energy and cost minimization in cloud computing systems.  
The input is a set of jobs $J=\{1,\dots,n\}$ where each job $j\in J$ has a release time $r_j$, deadline $d_j$, and processing time $p_j$. 
$m$ homogeneous machines are given with a parallelism parameter $g\geq 1$, which is the maximal number of jobs that can be processed simultaneously on a machine. 
A machine is called \emph{busy} when at least one job is being processed. 
The objective is to find a feasible schedule for all jobs such that the sum of busy times over all machines is minimized. 
We consider the online setting, where a job $j\in J$ is revealed at its release time $r_j$. 

We show multiple algorithms in different problem variants that have a tight competitive ratio. 
For the busy time scheduling problem, uniform processing time jobs, and where the parallelism is unbounded ($g=\infty$), we show a $2$-competitive algorithm and an online adversary that shows that the algorithm is tight. 
For the setting where jobs have arbitrary processing time, agreeable deadlines, and the parallelism is unbounded, we show a different tight $2$-competitive algorithm. 
For machines with bounded parallelism, we show lower bounds on the competitive ratio of any online algorithm when $g$ is small. 
Furthermore, we improve the setting with arbitrary jobs where the algorithm is allowed lookahead. 
\end{abstract}
\section{Introduction}\label{sec:introduction}
Busy time scheduling is a fundamental scheduling problem with applications to cloud computing and energy-aware systems. 
In cloud computing environments, customers pay for the usage time of a virtual machine, independent of the direct workload on the machine. 
The busy time scheduling model captures this setting. 

In the busy time scheduling problem, we are given a set of jobs $J=\{1,\dots,n\}$ where each job $j\in J$ has release time $r_j\in\mathbb{R}$, deadline $d_j\in\mathbb{R}$, and processing time $p_j\in\mathbb{R}$. 
A job $j\in J$ is \emph{rigid} iff $p_j = d_j-r_j$ and \emph{flexible} if $p_j\leq d_j-r_j$. 
Given are $m$ homogeneous machines, each able to process at most $g$ jobs in parallel. 
A feasible schedule assigns each job $j\in J$ to a machine in a non-preemptive interval with start time $s_j\in [r_j,d_j-p_j)$. 
For this interval $[s_j,s_j+p_j)$, the machine is \emph{busy}. 
Let $\mathcal{I}$ be the set of disjoint intervals $I\in\mathcal{I}$ where $I=[I^-,I^+)$ on a real line. 
The \emph{span} of $\mathcal{I}$ is defined as $Sp(\mathcal{I}) =\sum_{I\in\mathcal{I}}(I^+-I^-)$, i.e. the span $Sp$ is the Lebesgue measure. 
The span of a machine $i$ is $Sp(\cup_{j:m_j=i}[s_j,s_j+p_j))$, i.e. the total time at least one job is scheduled on machine $i$. 
The \emph{busy time} of a schedule is the sum of all machine spans: 
\[
    \sum_{i=1}^m Sp\left(\bigcup_{j : m_j = i}[s_j,s_j+p_j)\right)
\]

We study the online busy time scheduling problem where each job $j\in J$ is only revealed at its release time $r_j$. 
A job $j\in J$ must be completed within $[r_j,d_j)$. 
We use the standard competitive analysis for online algorithms against an all-powerful offline adversary on a worst-case input. 
An unbounded number of machines can be used if needed, as is the case for users of cloud computing environments. 
Busy time scheduling is an \emph{energy efficiency} problem as well. 
One approach to reduce energy usage is to minimize the time that a machine is enabled. 
Algorithms that optimize for energy efficiency have been extensively studied in the past \cite{DBLP:journals/cacm/Albers10,DBLP:journals/jacm/BansalKP07,DBLP:journals/scheduling/GerardsHH16,DBLP:journals/sigact/IraniP05,DBLP:conf/birthday/ChauL20}. 
Winkler et al. first showed that busy time scheduling is NP-hard for $g=2$ with rigid jobs \cite{DBLP:conf/soda/WinklerZ03}. 
They considered an optical network design problem to minimize the switching cost of the network, which is equivalent to the busy time scheduling problem with rigid jobs. 
Flammini et al. showed that a constant $4$-approximation is possible for the rigid jobs with arbitrary processing times \cite{DBLP:journals/tcs/FlamminiMMSSTZ10}. 
This was subsequently improved to a $3$-approximation by Shalom et al., who showed that no online algorithm has an absolute competitive ratio better than $g$ for the busy time scheduling problem \cite{DBLP:journals/tcs/ShalomVWYZ14}. 

For the offline busy time scheduling problem with flexible jobs, Khandekar et al. showed an optimal polynomial time algorithm when machines have unbounded parallelism ($g=\infty$). 
When machines have unbounded parallelism, the busy time scheduling problem reduces to minimizing the span of a single machine. 
Building on this, Chang et al. \cite{DBLP:journals/scheduling/ChangKM17} showed a $3$-approximation for flexible jobs, bounded parallelism ($g<\infty$), and arbitrary processing times. 
Fong et al. showed an optimal $O(n^3)$ dynamic program when jobs have \emph{agreeable} deadlines ($r_i\leq r_j$ iff $d_i\leq d_j$) and machines have unbounded parallelism. 
Khandekar et al. showed that the offline agreeable setting with rigid jobs and bounded parallelism has a greedy $2$-approximation  \cite{DBLP:conf/fsttcs/KhandekarSST10}. 

For the online busy time scheduling problem with flexible jobs, Ren and Tang showed two algorithms for unbounded parallelism, achieving $11.9$- and $6.9$-competitive ratios \cite{DBLP:conf/spaa/RenT17}. 
Simultaneously, Koehler and Khuller showed a $5$-competitive algorithm for arbitrary jobs and unbounded parallelism \cite{DBLP:conf/wads/KoehlerK17}. 
Both works showed that no online algorithm exists with a competitive ratio better than $\varphi=\frac{1+\sqrt{5}}{2}$ \cite{DBLP:conf/spaa/RenT17,DBLP:conf/wads/KoehlerK17}. 
Recently, Liu and Tang improved the lower bound for unbounded parallelism and arbitrary processing times to $4$ and noted flaws in the analysis of the claimed $4$-competitive algorithm of Fong et al. \cite{DBLP:conf/spaa/LiuT24,DBLP:conf/tamc/FongLLPWZ17}. 
Recent literature studied variants with learned predictions for the online setting with unbounded parallelism. 
Bovenkamp and Liu showed an online algorithm with access to a prediction on the time the machine should be active \cite{DBLP:conf/sofsem/BovenkampL25}. 
Liu and Tang considered the setting where the processing time of a job is provided as a prediction \cite{DBLP:conf/spaa/LiuT24}. 
When parallelism is bounded, Koehler and Khuller showed an online $12$-competitive algorithm with $2p_{max}$ lookahead in the setting with flexible jobs and arbitrary processing times, where $p_{max} = \max_{j\in J}(p_j)$ \cite{DBLP:conf/wads/KoehlerK17}. 
If the online algorithm is given a lookahead of $\ell$, any job $j\in J$ is revealed at time $r_j-\ell$. 
Lookahead weakens the adversary to overcome the lower bound on the competitive ratio of $g$ shown by Shalom et al. \cite{DBLP:journals/tcs/ShalomVWYZ14}. 
Recently, Calinescu et al. showed an $8$-competitive algorithm for unit jobs, and a tight $2$-competitive algorithm when the unit jobs have agreeable deadlines on heterogeneous machines \cite{DBLP:conf/esa/CalinescuDKZ24}.

\subsection{Our contributions}
In this paper, we study online busy time scheduling with flexible jobs, i.e., for any job \(j \in J\) it holds that \(p_j \le d_j - r_j\). 
The paper is split into two parts: unbounded parallelism (Section \ref{sec:non_preemptive_setting}) and bounded parallelism (Section \ref{sec:busy_time}).

For machines with unbounded parallelism we prove a lower bound of \(2-\epsilon\) for any online algorithm on uniform jobs, even under agreeable deadlines (Theorem \ref{thm:unbounded_uniform_lb}), and match it with a class of \(2\)-competitive algorithms (Theorem \ref{thm:unbounded_uniform}). 
For agreeable deadlines with arbitrary processing times, we give a \(2\)-competitive online algorithm, which is tight (Theorem \ref{thm:unbounded_agreeable}, Corollary \ref{col:unbounded_agreeable_lb}). 
To the best of our knowledge, this is the first tight bound for jobs with arbitrary processing times in the online setting.

For machines with bounded parallelism, we present a \(2\)-competitive algorithm for any \(g<\infty\) for uniform jobs (Theorem \ref{thm:bounded_uniform}) and present a lower bound for small \(g\) (Lemma~\ref{lem:bounded_uniform_lb}). 
For arbitrary processing times and \(p_{\max}\) lookahead, we obtain a \(9\)-competitive algorithm (Theorem \ref{thm:bounded_general_ub}), improving the algorithm by Koehler and Khuller \cite{DBLP:conf/wads/KoehlerK17} by halving the lookahead and tightening the ratio.

\section{Unbounded parallelism}
\label{sec:non_preemptive_setting}
We consider two natural restricted settings for machines with unbounded parallelism ($g=\infty$) and obtain tight results for both. 
In Section~\ref{sec:unbounded_uniform} we show a tight $2$-competitive algorithm for the fundamental variant where jobs have uniform processing times. 
In Section~\ref{sec:non_preemptive_agreeable} we show a tight $2$-competitive algorithm for the setting where jobs have agreeable deadlines.

\subsection{Uniform jobs}\label{sec:unbounded_uniform}
We consider the busy time scheduling problem with uniform processing times, i.e. $p_j=p$ for all $j\in J$. 
We assume without loss of generality that $p=1$. 
We show that any online algorithm $\mathcal{A}$ must have a competitive ratio of at least $2$ using an adaptive adversary. 
The adversary releases jobs in \emph{components} $C_i$, where $C_i$ is a set of jobs that any algorithm schedules within a contiguous busy interval on a single machine. 
Component $C_1$ has an initial rigid \emph{flag job} $f_1$ with $r_{f_1}=0$ and $d_{f_1}=1$, which has to be scheduled immediately. 
Consider component $C_i$ for $1\leq i\leq k$ for a large $k$. 
Let job $j$ be the latest job scheduled by algorithm $\mathcal{A}$. 
\begin{itemize}
    \item If $\mathcal{A}$ schedules job $j$ such that $s_j< d_{f_i}$, release job $j'$ at $r_{j'}=s_j+\epsilon$ with $d_{j'}=i3$. 
    \item If $\mathcal{A}$ delays job $j$ such that $s_j \geq d_{f_i}$, mark job $j$ as flag job $f_{i+1}$ and jobs $j'$ with $r_{j'}>r_j$ are part of component $C_{i+1}$. 
    Release job $j'$ at $r_{j'}=s_{f_{i+1}}+\epsilon$ with deadline $d_{j'}=3(i+1)$. 
\end{itemize}
The adversary is defined such that jobs in every other component in the schedule can be rearranged, which almost removes the busy time cost of the rearranged components. 
For ease of notation, we say the infimum of $C_i$ is defined as $\inf(C_i) = \min\{s_j : j\in C_i\}$ and the supremum of $C_i$ is defined as $\sup(C_i) = \max\{s_j+1 : j\in C_i\}$.
%
%
Before formalizing this in Theorem~\ref{thm:unbounded_uniform_lb}, we show the following lemma, which is shown using a simple averaging argument. 

\begin{lemma}\label{lem:unbounded_uniform_lb_independent_set}
    Given are values $X=\{x_1,\dots,x_k\}$. 
    There exists a set $S\subseteq X$ of non-consecutive $x_i\in S$ such that $\sum_{x_i\in S}x_i\geq\frac{1}{2}\sum_{x_i\in X}x_i$. 
\end{lemma}

Let the span of component $C_i$ be denoted as $x_i =  \sup(C_i) - \inf(C_i)$. 
For component $C^*_i$ in the optimal schedule $S^*$, its span is denoted as $x^*_i$. 

\begin{theorem}\label{thm:unbounded_uniform_lb}
    No online algorithm exists with a competitive ratio less than $2$ for the busy time scheduling problem with unbounded parallelism and uniform jobs. 
\end{theorem}
\begin{proof}
    We consider an arbitrary online algorithm $\mathcal{A}$ with components $C_1,\dots,C_k$ for some large $k$ according to the adversary and an optimal schedule $S^*$. 
    Consider any three consecutive components $C_{i-1},C_{i},C_{i+1}$. 
    The flag job $f_i$ has release time $r_{f_i} = \sup(C_{i-1})-1+\epsilon$. 
    Therefore, flag job $f_i$ can be rescheduled to overlap with the jobs from component $C_{i-1}$ at $s_{f_i}=r_{f_i}$, which adds $\epsilon$ busy time to component $C_{i-1}$. 
    All other jobs $C_i\backslash\{f_i\}$ have deadline $\inf(C_{i+1})+1$ and can therefore be scheduled with component $C_{i+1}$. 
    Therefore, rescheduling all jobs in component $C_i$ decreases the busy time by $x_i-\epsilon$. 
    Let $X$ be the set of all $x_1,\dots,x_k$ of the respective components $C_1,\dots,C_k$. 
    Let $X'\leftarrow X\backslash\{x_1,x_k\}$, since both $C_1$ and $C_k$ do not have two adjacent components. 
    Since all components in $X'$ have two adjacent components, we apply Lemma~\ref{lem:unbounded_uniform_lb_independent_set} to obtain set $S'$ of non-consecutive components from $X'$ such that $\sum_{x_i\in S'}x_i\geq \frac{1}{2}\sum_{x_i\in X'}x_i$. 
    Let $S\leftarrow X'\backslash S'$ be the set of remaining components with $\sum_{x_i\in S}x_i\leq\frac{1}{2}\sum_{x_i\in X'}x_i$. 
    Therefore, $ALG= x_1+x_k+\sum_{x_i\in X'}x_i\geq x_1+x_k+\epsilon|S|+2\sum_{x_i\in S}x_i$ where $\epsilon|S|$ is the additional busy time per component after rescheduling. 
    This gives the following bound for any algorithm $\mathcal{A}$. 
    \[
        x_1+x_k+\epsilon|S|+2\sum_{x_i\in S}x_i\geq \epsilon|S|-x_1-x_k+2\sum_{x_i\in S^*}x_i= 2OPT+\epsilon|S|-x_1-x_k
    \]
    where the inequality holds by adding the omitted $x_1,x_k$ back to $S$ to consider all components, which is an upper bound on the optimal schedule.  
    Since $x_1$ and $x_k$ are bounded by a constant, picking a sufficiently small $\epsilon \ll 1/k$, e.g. $\epsilon:=1/k^2$ bounds $\epsilon|S|$ to $1/k$, and $\sum_{x_i\in X}x_i$ and $\sum_{x_i\in S}x_i$ are arbitrarily large for arbitrarily large $k$, it holds that $ALG\ge2OPT+\epsilon|S|-x_1-x_k = (2-\varepsilon)OPT$ for an arbitrarily small $\varepsilon$ depending on $k$. 
    Since $\varepsilon$ only depends on an arbitrary $k$, no online algorithm exists with a competitive ratio less than $2$. 
\end{proof}

We show a $2$-competitive algorithm for unbounded parallelism and uniform jobs. 
Many algorithms in the literature use the scheduling scheme of waiting until a job is at its deadline (the \emph{flag} job) and scheduling jobs alongside these flag jobs for some interval $[s_j,d_j+\alpha)$ \cite{DBLP:conf/tamc/FongLLPWZ17,DBLP:conf/wads/KoehlerK17,DBLP:conf/esa/CalinescuDKZ24,DBLP:conf/spaa/RenT17}. 
We show that for any $\alpha\in[0,1)$, this class of algorithms is $2$-competitive. 
The algorithm \alg{UnboundedUniform} with input parameter $\alpha$ therefore represents this entire class of algorithms. 
Algorithm \alg{UnboundedUniform} maintains a set of flag jobs $F$ and a set of intervals $S$ representing the schedule. 
\begin{algorithm}[hbt!]
\caption{\alg{UnboundedUniform}$(\alpha)$}\label{alg:unbounded_uniform_scheduler}
\DontPrintSemicolon
$S,J,F\leftarrow\emptyset$\tcp*{schedule, pending jobs, flag jobs}
\For{$t=0$ to $d_{\max}$ with $t\in\mathbb{R}$}{
    Add newly released jobs at $t$ to $J$\;
    Schedule all jobs $i\in J$ for which $[t,t+1)\subseteq S$; Remove all $i$ from $J$\;
    \If{$t = d_j-1$ for some $j\in J$}{
        $F\leftarrow F\cup\{j\}$;
        $S\leftarrow S\cup [t,t+1+\alpha)$\;
        Schedule all $i\in J$ with $[t,t+1)\subseteq S$; Remove all $i$ from $J$\;
    }
}
\end{algorithm}

\begin{theorem}\label{thm:unbounded_uniform}
    The online busy time scheduling problem with unbounded parallelism and uniform jobs admits a tight $2$-competitive algorithm. 
\end{theorem}
\begin{proof}
    For each flag job $f\in F$, it holds that $s_f = d_f-1$ and uses interval $[s_f, d_f+\alpha)$. 
    Consider the connected components $C^*_1,\dots,C^*_k$ of scheduled jobs that partition the optimal schedule. 
    Let $F_i\subseteq F$ such that $F_i\subseteq C^*_i$, i.e. all flag jobs scheduled in component $C^*_i$. 
    Since jobs are uniform, for any two flag jobs $f_{j},f_{j+1}$ it holds that $s_{f_j}+\alpha<r_{f_{j+1}}$ otherwise job $f_{j+1}$ would not be a flag job by the definition of \alg{UnboundedUniform($\alpha$)}. 
    We consider the flag jobs $\{f_p,\dots,f_q\}= F_i$. 
    Since job $f_q\in F_i$ and $f_p$ is scheduled at its deadline, it holds that $[s_{f_p},d_{f_{q-1}}+\alpha)\subseteq[\inf(C^*_i),\sup(C^*_i))$. 
    When $|F_i|\geq2$, then the last flag job $f_q$ contributes at most $1+\alpha$ which is less than $[s_{f_p},d_{f_{q-1}}+\alpha)$. 
    If $|F_i| = 1$, it holds that $Sp(C^*_i)\geq 1$, and job $f_p\in F_i$ contributes at most $1+\alpha$. 
    Since $\alpha\leq 1$, we obtain the following bound for all components. 
    \[
        ALG\leq\sum_{f\in F}(1+\alpha) \leq 2\sum_{i: |F_i|\geq 2}Sp(C^*_i) + 2\sum_{i: |F_i|=1}Sp(C^*_i)=  2\cdot OPT
    \]
    Combining the lower bound with the above upper bound gives the theorem. 
\end{proof}

\subsection{Agreeable setting}\label{sec:non_preemptive_agreeable}
A job instance $J$ is agreeable if it holds that $r_i\leq r_j$ iff $d_i\leq d_j$ for any two jobs $i,j\in J$. 
Calinescu et al. considered unit jobs with agreeable deadlines in the online setting with heterogeneous machines \cite{DBLP:conf/esa/CalinescuDKZ24} and obtained a tight $2$-competitive algorithm. 
Their lower bound of $2$ relies on the existence of heterogeneous machines and therefore does not apply to the homogeneous setting. 
We consider the setting where machines have unbounded parallelism and show a tight $2$-competitive algorithm. 
The job instance created by the adversary in Theorem~\ref{thm:unbounded_uniform_lb} is agreeable. 
\begin{corollary}\label{col:unbounded_agreeable_lb}
    No online algorithm exists with a competitive ratio less than $2$ for the busy time scheduling problem with unbounded parallelism, uniform jobs, and agreeable deadlines. 
\end{corollary}
We use the following structural lemma, shown by Fong et al. \cite{DBLP:conf/tamc/FongLLPWZ17}, for offline agreeable busy time scheduling with unbounded parallelism.  

\begin{lemma}[\cite{DBLP:conf/tamc/FongLLPWZ17}]\label{lem:unbounded_agreeable_ordered_jobs}
    For jobs with agreeable deadlines, an optimal schedule exists such that the jobs are processed in the order of their release times. 
\end{lemma}

We consider the following online algorithm. 
The algorithm waits until a job $f\in J$ is at its starting deadline $d_f-p_f$, breaking ties arbitrarily. 
Job $f$ is marked as a \emph{flag job} by adding $f$ to $F$, and job $f$ is started at $s_f=d_f-p_f$ which opens the interval $[s_f,d_f)$. 
Any job $j\in J$ with $r_j\leq d_f$ is started as early as possible within $[s_f,d_f)$, which adds $[s_j,s_j+p_j)$ to the schedule $S\leftarrow S\cup[s_j,s_j+p_j)$. 
Any job $j$ for which $[r_j,r_j+p_j)\subseteq S$ is scheduled immediately which does not add any busy time. 
The algorithm maintains the invariant that flag jobs $F$ are scheduled pair-wise disjoint. 
The algorithm is described in Algorithm~\ref{alg:unbounded_agreeable_scheduler}. 
For the analysis, we consider the connected components $C_1,\dots,C_k$. 
Component $C_i$ contains a unique flag job $f_i$ which starts at ($\inf(C_i) = s_{f_i}$), and jobs $j$ with $r_j\in[s_{f_i},s_{f_i}+p_{f_i})$, which are started within $s_j\in[s_{f_i},d_{f_i})$. 

\begin{algorithm}[hbt!]
\caption{\alg{UnboundedAgreeable}}\label{alg:unbounded_agreeable_scheduler}
\DontPrintSemicolon
$S,J,F\leftarrow\emptyset$\tcp*{schedule, pending jobs, flag jobs}
\For{$t=0$ to $d_{\max}$ with $t\in\mathbb{R}$}{
    Add newly released jobs at $t$ to $J$\;
    \If{$t\not\in \cup_{f\in F}[s_f,s_f+p_f)$ and $J\neq\emptyset$ and $t=d_j-p_j$ for some $j\in J$}{
        $s_j\leftarrow t$;
        $F\leftarrow F\cup \{j\}$;
        $S\leftarrow S\cup[s_j,s_j+p_j)$;\;
    }
    \For{$j\in J$}{
        \If{$t\in\cup_{f\in F}[s_f,s_f+p_f)$ or $[t,t+p_j)\subseteq S$}{
            $s_j\leftarrow t$; 
            $S\leftarrow S\cup[s_j,s_j+p_j)$; 
            $J\leftarrow J\setminus{j}$;
        }
    }
}
\end{algorithm}


\begin{lemma}\label{lem:unbounded_agreeable_consecutive_flag_jobs}
    Consider all components $C_i,\dots,C_j$ with their respective flag jobs $F'=\{f_i,\dots,f_j\}$ such that $F'\subseteq C^*_m$ in the optimal schedule. 
    It holds that $[\inf(C_i),\sup( C_{j-1}))\subseteq[\inf(C^*_m),\sup( C^*_m))$. 
\end{lemma}
\begin{proof}
    Since the flag job $f_i\in C_i$ is scheduled at its deadline by the definition of the flag job, it holds that $\inf(C_i)\geq\inf(C^*_m)$. 
    Since the flag job $f_j\in C_j$ did not schedule with $C_{j-1}$ it holds that $d_{f_{j-1}}<r_{f_j}$ for flag job $f_{j-1}\in C_{j-1}$, and $r_{f_j}+p_{f_j} >\sup(C_{j-1})$. 
    Since flag job $f_j$ is part of $C^*_m$, this means that $\sup(C^*_m)\geq r_{f_j}+p_{f_j}\geq\sup(C_{j-1})$. 
    Combining this gives $[\inf(C_i),\sup(C_{j-1}))\subseteq[\inf(C^*_m),\sup(C^*_m))$.
\end{proof}

We now show that the algorithm \alg{UnboundedAgreeable} has a competitive ratio of $2$. 
The proof uses a charging argument over the optimal connected components $C^*_1,\dots,C^*_k$. 

\begin{lemma}\label{lem:unbounded_agreeable_alg_2_competitive}
    Algorithm~\alg{UnboundedAgreeable} is $2$-competitive for busy time scheduling with unbounded parallelism and agreeable job instances. 
\end{lemma}
\begin{proof}
    Consider an optimal schedule with jobs scheduled according to Lemma~\ref{lem:unbounded_agreeable_ordered_jobs}. 
    Let $C^*_1,\dots,C^*_k$ be the components that partition the optimal schedule where each component contains a set of jobs that overlap with spans $x^*_1,\dots,x^*_k$.  
    Algorithm~\ref{alg:unbounded_agreeable_scheduler} creates flag jobs $f_1,\dots,f_m$. 
    Consider the first component $C^*_1$ in the optimal schedule and let $F_1=\{f_1,\dots,f_j\}$ be all flag jobs $F_1\subseteq F$ where $F_1\subseteq C^*_1$. 
    By Lemma~\ref{lem:unbounded_agreeable_consecutive_flag_jobs}, the total span used by all jobs scheduled with flag jobs $f_1,\dots,f_{j-1}$ is at most $x^*_1$, which we charge to $C^*_1$. 
    For the last flag job $f_j$, we consider job $j'$ scheduled with $f_j$ that maximizes $r_{j'}+p_{j'}$. 
    If job $j'\in C^*_1$, then $r_{j'}+p_{j'}\leq\sup(C^*_1)$, so it contributes at most $x^*_1$ which we charge to component $C^*_1$. 
    Otherwise, job $j'\not\in C^*_1$, we only charge $[s_{f_j},d_{f_j})$ to component $C^*_1$ and defer  $[d_{f_j},r_i+p_i)$ to the next optimal component. 
    Consider the component $C^*_i$ and let $f_p,\dots,f_q$ be the flag jobs scheduled in $C^*_i$. 
    We charge three parts to $C^*_i$. 
    \begin{enumerate}
        \item The deferred tail from $f_{p-1}$, i.e. for the job $j'\in C^*_{i}$ scheduled with $f_{p-1}$ that maximizes $r_{j'}+p_{j'}$, we charge $[d_{f_{p-1}},r_{j'}+p_{j'})$ to $[\inf(C^*_i),d_{f_p})$.
        Since $j'\in C^*_i$ and jobs have agreeable deadlines, $d_{j'}\leq d_{f_p}$ and $\inf(C^*_i)\leq d_{j'}-p_{j'}$. 
        \item All jobs scheduled with jobs $f_p,\dots,f_{q-1}$ are at most $x^*_i$ by Lemma~\ref{lem:unbounded_agreeable_consecutive_flag_jobs}, which we charge to $C^*_i$. 
        \item Let job $j'$ scheduled with $f_p$ maximize $r_{j'}+p_{j'}$ over all jobs scheduled with $f_p$. 
        If job $j'\in C^*_i$, then $r_{j'}+p_{j'}\leq\sup(C^*_i)$ and charge $[s_{f_q},r_{j'}+p_{j'})$ to $C^*_i$. 
        If job $j'\not\in C^*_i$, then charge $[s_{f_q},d_{f_q})$ to $[d_{f_{q-1}},\sup(C^*_i))$ and defer $[d_{f_q},r_{j'}+p_{j'})$ to the next component. 
    \end{enumerate}
    For the first and third cases, we charge $d_{f_p} - \inf(C^*_i) + \sup(C^*_i) - d_{f_{q-1}}$ to $C^*_i$. 
    Since $d_{f_{q-1}}\geq d_{f_p}$, it holds that $d_{f_p} - \inf(C^*_i) + \sup(C^*_i) - d_{f_{q-1}}\leq \sup(C^*_i) - \inf(C^*_i) = x^*_i$. 
    The second case charges at most $x^*_i$ onto $C^*_i$. 
    Therefore, each optimal component $C^*_i$ receives a total charge of at most $2\cdot x^*_i$. 
    Then $ALG \leq \sum_{i}2\cdot x^*_i = 2\cdot OPT$, which proves the lemma. 
\end{proof}

By Corollary~\ref{col:unbounded_agreeable_lb}, no online algorithm exists with a competitive ratio less than $2-\varepsilon$ on agreeable instances.
Combining this lower bound with Lemma~\ref{lem:unbounded_agreeable_alg_2_competitive} gives the following theorem. 

\begin{theorem}\label{thm:unbounded_agreeable}
    The online busy time scheduling problem with unbounded parallelism and agreeable jobs admits a $2$-competitive tight algorithm. 
\end{theorem}

\section{Bounded parallelism}\label{sec:busy_time}
We study homogeneous machines, each with bounded parallelism $g$, i.e. at most $g$ jobs may run simultaneously on the same machine. 
An unbounded number of machines can be opened if necessary. 
The cost of a schedule is the sum of machine spans. 
For heterogeneous machines, Calinescu showed an $8$-competitive algorithm for uniform jobs \cite{DBLP:conf/esa/CalinescuDKZ24}. 
For uniform jobs on homogeneous machines, we give a lower bound for small $g=2$, and a $2$-competitive algorithm without the need for lookahead. 
For arbitrary job processing times, we improve on the work of Koehler and Khuller \cite{DBLP:conf/wads/KoehlerK17}. 

\subsection{Online busy time with uniform jobs}\label{sec:busy_time_uniform}
We consider jobs $j\in J$ with uniform processing time, i.e., $p_j=p$ for all $j\in J$, and assume that $p_j = 1$. 
Since any schedule is optimal when machines have parallelism of $g=1$, we assume that $g\geq 2$. 
In Section~\ref{sec:unbounded_uniform} we showed that no online algorithm exists that admits a competitive ratio less than $2$. 
This lower bound uses many jobs scheduled in parallel on the same machine.  
We show that no algorithm has a competitive ratio less than $\sqrt{2}$, even if $g=2$. 

\begin{lemma}\label{lem:bounded_uniform_lb}
    No online algorithm with a competitive ratio less than $\sqrt{2}$ exists for busy time scheduling with bounded parallelism $g\geq 2$ and uniform jobs. 
\end{lemma}
\begin{proof}
It is sufficient to prove the bound for $g=2$. 
Let $\alpha\in[0,1]$ which we fix later. 
At time $0$, release two jobs $j_1,j_2$ where $d_1 = 1$ and $d_2 = 3$. 
The adversary releases the next jobs adaptively in the following two cases. 
\begin{enumerate}
    \item If job $j_2$ is scheduled after $\alpha$, the adversary releases no more jobs. 
    Only a single machine is used, which has a span of at least $1+\alpha$. 
    \item If job $j_2$ is scheduled before or at $\alpha$, release rigid job $j_3$ at $r_3=s_2\leq\alpha$ with deadline $d_3=r_3+1$, which has to use a new machine. 
    The adversary releases job $j_4$ with $r_{j_4}=2$ and $d_{j_4}=3$. 
    Jobs $j_1$, $j_2$, and $j_4$ use a span of $2+s_2$. 
    Since job $j_3$ is scheduled on a new machine, it uses a span of $1$.
    Therefore, any algorithm incurs a cost of $3+s_2$. 
    Optimally, job $j_2$ is started at $s_2=r_4$ with job $j_4$, and jobs $j_1$ and $j_3$ are scheduled on the same machine, which incurs a busy time cost of $2+s_3$. 
\end{enumerate}
The competitive ratio of any algorithm is at least $\min\bigl\{\frac{1+\alpha}{1},\frac{3+s_2}{2+s_2}\bigr\}$. 
The competitive ratio is minimized when $s_2=\alpha$. 
For $\alpha=\sqrt{2}-1$, this gives $\frac{3+\alpha}{2+\alpha}= \frac{1+\alpha}{1}$, which gives a competitive ratio lower bound of $1+\alpha = \sqrt{2}$ for $g=2$. 
\end{proof}

We give a $2$-competitive online algorithm for homogeneous machines with bounded parallelism $g$ and uniform jobs. 
We use \emph{bundles} to group jobs together with parallelism at most $g$. 
Each bundle $B_i$ is scheduled on a new machine by the algorithm. 
When the earliest deadline job $j$ must be scheduled at $s_j=d_j-1$, adds it to a new bundle. 
The bundle is greedily filled to capacity with jobs that fit within $[s_j,s_j+2)$, prioritizing jobs with earlier deadlines. 
The algorithm is given in Algorithm~\ref{alg:busy_time_uniform_scheduler}. 

\begin{algorithm}[hbt!]
\caption{\alg{UniformScheduler}}\label{alg:busy_time_uniform_scheduler}
\DontPrintSemicolon
$S,J,F,\mathcal{B}\leftarrow\emptyset$\tcp*{schedule, pending jobs, flag jobs, bundles}
\For{$t=0$ to $d_{max}$ with $t\in\mathbb{R}$}{
    Add newly released jobs at $t$ to $J$\;
    \If{$d_f\leq t$ for some $f\in F$ and less than $g$ jobs are assigned to $B_i\in\mathcal{B}$}{
        \eIf{$d_f<t$ and there exists $j\in J$ with $d_j-1\leq d_f$}{
        Schedule $j$ to $B_i$; $J\leftarrow J\backslash\{j\}$;\;
        }{
        Fill $B_i$ with up to $g$ jobs $j\in J$ by non-decreasing $d_j$; $J\leftarrow J\backslash\{j\}$;\;
        }
    }
    \If{There exists a job $j\in J$ such that $t=d_j-p$}{
        $F = F\cup\{j\}$;
        Create bundle $B$ with $j\in B$ and $s_j=t$; 
        $\mathcal{B}\leftarrow\mathcal{B}\cup\{B\}$\;
    }
}
Schedule all $\mathcal{B}$ to $S$ on separate machines\;
\end{algorithm}

\begin{theorem}\label{thm:bounded_uniform}
     Online busy time scheduling with uniform processing time jobs admits a $2$-competitive algorithm. 
\end{theorem}
\begin{proof}
    Let $B_i\in\{B_1,\dots,B_{k'}\}$ be the \emph{bundles} created by the algorithm in order, where each bundle $B_i\subseteq J$ is a set of jobs scheduled on the same machine with its unique flag job $f_i\in F$ with $s_{f_i}=d_{f_i}-1$. 
    Each bundle $B_i$ occupies $[s_{f_i},s_{f_i}+2)$, which contributes span at most $2$. 
    Hence $ALG\leq\sum^{k'}_{i=1}2 = 2|F|$, since $p=1$ for all jobs in $F$. 

    We define a decomposition of an optimal schedule. 
    Take the maximal set of pairwise unit intervals $I^*_i=[t_i,t_i+1)$ over the span of the machines. 
    Let optimal bundle $B^*_i$ be the set of all jobs such that $s_j\in[t_i,t_i+1)$ where $f^*_i$ is the flag job with the earliest start time $s_{f^*_i}$ in $B^*_i$. 
    Let $F^*$ be the set of all such flag jobs. 
    Since intervals are disjoint, it holds that $k:=|F^*|\leq OPT$ and $|B^*_i|\leq g$. 
    We index the bundles such that $d_{f^*_1}\leq\dots\leq d_{f^*_k}$. 
    
    We show that $k'\leq k$. 
    We show by induction that $|B^*_1\cup\dots\cup B^*_i|\leq|B_1\cup\dots\cup B_i|$. 
    The base case $i=1$ holds by construction. 
    For the inductive step, consider a job $j\in B^*_i$ such that $j\not\in\cup_{1\leq i'\leq i}B_{i'}$, i.e. not scheduled by the algorithm. 
    The algorithm chooses job $f_i$ with the smallest remaining deadline, so $t_i\leq d_{f_i}$. 
    Since the algorithm does not schedule job $j$, $d_j\geq d_{f_i}$ and $s_{f_i}+1= d_{f_i}\leq d_j$. 
    Since $r_j\leq t_i+1\leq d_{f_i}+1$ as $j\in B^*_i$, job $j$ can be assigned to $B_i$ within the interval $[s_{f_i},s_{f_i}+2)$. 
    Moreover, it holds that $|B^*_i|\leq g$ and $B_i$ has capacity $2g$ across interval $[s_{f_i},s_{f_i}+2)$. 
    Therefore, if $|B^*_1\cup\dots\cup B^*_i|>|B_1\cup\dots\cup B_i|$, job $j\in B^*_i$ is unassigned, but feasible in $B_i$ which gives a contradiction. 
    Therefore $k'\leq k$. 
    \[
        ALG\leq2|F|\leq2|F^*|\leq 2\cdot OPT
    \]
    where $|F|\leq|F^*|$ holds by $k'\leq k$.
    This gives a $2$-competitive ratio for the online algorithm. 
\end{proof}

\subsection{Online busy time scheduling with arbitrary processing times}\label{sec:busy_time_general_lookahead}
We consider arbitrary processing times with homogeneous machines of bounded parallelism $g$. 
We allow the online algorithm \emph{lookahead}. 
With a lookahead of $\ell$, job $j\in J$ is revealed at time $r_j-\ell$. 
Shalom et al. showed no online algorithm can achieve a competitive ratio better than $g$ in this setting \cite{DBLP:journals/tcs/ShalomVWYZ14}. 
Koehler and Khuller obtained a $12$-competitive algorithm with $2p_{max}$ lookahead where $p_{max}= \max_{j\in J}(p_j)$ \cite{DBLP:conf/wads/KoehlerK17}. 
We improve this to a $9$-competitive online algorithm with only $p_{max}$ lookahead. 
Following Koehler and Khuller \cite{DBLP:conf/wads/KoehlerK17}, we first fix jobs $j\in J$ to an interval within $[r_j,d_j)$ and treat the job as rigid. 
We assign the rigid jobs to machines using \emph{tracks} as in Chang et al. \cite{DBLP:journals/scheduling/ChangKM17}. 
\begin{definition}[\cite{DBLP:journals/scheduling/ChangKM17}]
    A \textbf{track} is a set of pairwise disjoint rigid jobs.
\end{definition}
Tracks are grouped into \emph{bundles} $B_i$, each containing at most $g$ tracks. 
Bundles are then feasibly assigned to individual machines. 
Since $g=1$ is trivial, we assume that $g\geq 2$. 
We show that $Sp(B_1)\leq OPT_\infty(U)\leq OPT(U)$ where $OPT_\infty$ is the optimal schedule of jobs $U$ of rigid jobs after fixing them for unbounded parallelism. 
Similar to Chang et al., we show that $\sum_{i>1}Sp(B_i)\leq\frac{4}{g}\ell(U)\leq 4OPT(U)$ where $\ell(U)$ is the \emph{load} ($\ell(U)=\sum_{j\in U}p_j$) of the jobs in set $U$ \cite{DBLP:journals/scheduling/ChangKM17}. 

We show a new machine assignment algorithm \alg{OnlineGreedyTracking} for the online setting given in Algorithm~\ref{alg:online_greedy_tracking}. 
Let $\mathcal{Q}_i=(T_{i,1},T_{i,2})$ be a pair of tracks. 
\alg{OnlineGreedyTracking} finds pairs of tracks $\mathcal{Q}_i$ that spans all other jobs assigned to $\mathcal{Q}_j$ for $j>i$. 
Consider the first pair of tracks $\mathcal{Q}_1$. 
Once $t=r_j$, the job is assigned to $T_{1,1}$ breaking ties in favor of later deadlines. 
Since we consider a lookahead of $p_{\max}$, we find the job $j'$ with the maximum deadline such that $r_{j'}\leq d_{j}$ and reserve it for $T_{1,2}$. 
If no such $j'$ exists, no job overlaps with $j$, and the job $j'$ with the earliest $r_{j'}$ is chosen. 
If any job $j''$ is revealed with $d_{j''} > d_{j'}$ and $r_{j''}\leq d_j$, job $j''$ replaces the reserved job $j'$. 
Once we reach the reserved job, it is scheduled and the process is repeated. 
Once the reserved job is scheduled, no revision has to be done because of the $p_{\max}$ lookahead. 
This is done in order of the tracks and index $i$ of pairs $\mathcal{Q}_i$. 
The pairs of tracks are assigned to bundles in order of their index. 
Up to $g$ tracks are assigned to a single bundle $B_k$. 

\begin{algorithm}[hbt!]
\caption{\alg{OnlineGreedyTracking}}\label{alg:online_greedy_tracking}
\DontPrintSemicolon
\For{$t=0$ to $d_{\max}$ with $t\in\mathbb{R}$}{
  $U\leftarrow\{j:\ r_j\le t+p_{\max}$ with $j$ unassigned and unreserved$\}$\;
  \For{all pairs $\mathcal{Q}_i$ in order}{
    \If{its reserved job $j$ has $r_j=t$}{place $j$; reserve $j'\in U$ with $\max d_{j'}$ overlapping $j$ (else earliest $r_{j'}$)}
    Reserve $j'\in U$ with $\max d_{j'}$ overlapping assigned job $j$ in $\mathcal{Q}_i$
  }
  \For{$j\in U$ with $r_j=t$ not assigned (breaking ties on $d_j$)}{
    start new pair $\mathcal{Q}$; place $j$ in $\mathcal{Q}$; pick $j'$ as above; reserve $j'$ to $\mathcal{Q}$
  }
}
\end{algorithm}

\begin{lemma}\label{lem:bounded_general_OnlineGreedyTracking}
    The online algorithm \alg{OnlineGreedyTracking} is $4$-competitive when granted $p_{max}$ lookahead. 
\end{lemma}
\begin{proof}
    We maintain the invariant that earlier pairs cover subsequent pairs. 
    Consider the pair of tracks $\mathcal{Q}_i$. 
    We only consider jobs that are neither assigned nor reserved to earlier tracks. 
    Whenever a job $j$ is revealed for the first time, the algorithm can schedule $j$ to a track of $\mathcal{Q}_i$. 
    Suppose job $j$ is not assigned or reserved to $\mathcal{Q}_i$. 
    We consider the following cases. 
    \begin{enumerate}
        \item If no job is assigned to the tracks of $\mathcal{Q}_i$, job $j$ has the earliest release time $r_j$ and would be assigned. 
        \item If at most one job $j'$ assigned to $\mathcal{Q}_i$ is not disjoint from job $j$, job $j$ can be reserved to $\mathcal{Q}_i$. 
        Since it is not, it is contained in the span of $\mathcal{Q}_i$, or there exists a job $j''$ reserved to $\mathcal{Q}_i$ with $d_{j''}\geq d_{j}$. 
        Since $j''$ was reserved to $\mathcal{Q}_i$ before $j$, it holds that $r_{j''}\leq r_j$. 
        Therefore, job $j''$ is not disjoint with $j'$, so $d_{j''}\geq d_j$ and $r_{j''}\leq r_j$. 
        Thus job $j$ is contained in the span of $j'$ and $j''$ and is not assigned to $\mathcal{Q}_i$.  
        Otherwise, substitute $j''$ for $j$ as the reserved job. 
    \end{enumerate}
    In all cases, when a new job $j$ is revealed, it can be reassigned or reserved to a track. 
    The reservations for all tracks is revised when a new job is revealed, maintaining the span invariant that earlier pairs cover subsequent pairs. 
    Once a job $j$ is scheduled, all jobs that overlap with $j$ are revealed by the $p_{\max}$ lookahead, so no revision is needed for $j$. 
    Let $\mathcal{Q}_i$ be the earliest pair assigned to bundle $B_k$. 
    Since $\mathcal{Q}_i$ covers all subsequent pairs, it holds that $Sp(\mathcal{Q}_i)=Sp(B_k)$. 
    Notice that $\frac{1}{g}\ell(J)$ is a lower bound on the span to schedule all jobs $J$ on a single machine. 
    We look at the total sum of processing times that is assigned to a bundle and charge it to the earlier bundle. 
    By the invariant it holds that $Sp(\mathcal{Q}_i)\leq Sp(\mathcal{Q}_{i-1})$ where pair $\mathcal{Q}_{i-1}$ is assigned to $B_{i-1}$. 
    Therefore the following inequality holds. 
    \[
        Sp(B_i)=Sp(\mathcal{Q}_i)\leq\ell(T_{i,1}\cup T_{i,2})\leq2\cdot\ell(T^*_{i,1}\cup T^*_{i,2})\leq2\cdot\frac{2}{g}\ell(B_{i-1})
    \]
    Where $\ell(T^*_{i,1}\cup T^*_{i,2})$ is the upper bound of the sum of processing times of the jobs assigned to the first two tracks of machine $i$. 
    With the additional factor $2$ due to the number of tracks used for each $\mathcal{Q}_i$ is at most $2$. 
    Therefore, $2\cdot\frac{2}{g}\ell(B_{i-1})\leq 4\cdot OPT(J)$ for the bundles $B_i$ with $i>1$ where $OPT(J)$ is the optimal schedule for the jobs $J$.
\end{proof}

We use the unbounded $5$-competitive algorithm by Koehler and Khuller \cite{DBLP:conf/wads/KoehlerK17} for the initial jobs to obtain the competitive ratio of $9$. 
\begin{theorem}\label{thm:bounded_general_ub}
    \alg{OnlineGreedyTracking} is a $9$-competitive algorithm for online busy time scheduling with $p_{\max}$ lookahead. 
\end{theorem}
\begin{proof}
    We get $Sp(B_1)\leq OPT_{\infty}(J')\leq 5\cdot OPT(J)$ obtained by Koehler and Khuller \cite{DBLP:conf/wads/KoehlerK17} which in combination with \alg{OnlineGreedyTracking} by Lemma~\ref{lem:bounded_general_OnlineGreedyTracking} gives an online ratio of 
    \[
        Sp(B_1)+\sum_{i>1}Sp(B_i)\leq 5\cdot OPT(J) + 4\cdot OPT(J) = 9\cdot OPT(J)
    \]
    This results in a $9$-competitive algorithm by fixing the jobs $p_{max}$ in advance using the $5$-competitive algorithm by Koehler and Khuller, and the \alg{OnlineGreedyTracking} algorithm. 
\end{proof}



\begin{credits}
\subsubsection{\ackname} This work was funded by the Deutsche Forschungsgemeinschaft (DFG, German Research Foundation) - GRK 2201 - Projektnummer 277991500. 
We thank the anonymous reviewers for their thoughtful and constructive feedback. 

\end{credits}

\end{document}